\documentclass{birkjour}
 \newtheorem{thm}{Theorem}[section]
 \newtheorem{cor}[thm]{Corollary}
 \newtheorem{lem}[thm]{Lemma}
 \newtheorem{prop}[thm]{Proposition}
 \theoremstyle{definition}
 \newtheorem{defn}[thm]{Definition}
 \theoremstyle{remark}
 \newtheorem{rem}[thm]{Remark}
 
 \numberwithin{equation}{section}

\usepackage{hyperref}
\begin{document}

%
%
%
%
%
%
%
%
%

\title{Transfer and entanglement stability of property ($UW${\normalsize\it{E}})}

\author{Sinan Qiu\textsuperscript{1}}
\address{1 School of Mathematics and Statistics,
  Beijing Institute of Technology,
  Beijing 100081, People's Republic of China}
\email{qiusinan111@163.com}

\author{Lining Jiang\textsuperscript{1${\ast}$}}
\thanks{*Corresponding author}
\address{1 School of Mathematics and Statistics,
  Beijing Institute of Technology,
  Beijing 100081, People's Republic of China}
\email{jianglining@bit.edu.cn}


\date{January 1, 2004}
\dedicatory{}
\thanks{This research did not receive any specific grant from funding agencies in the public, commercial, or not-for-profit sectors.}


\begin{abstract}
An operator $T\in B(H)$ is said to satisfy property ($UW${\scriptsize \it{E}}) if the complement in the approximate point spectrum of the essential approximate point spectrum
coincides with the isolated eigenvalues of the spectrum.
Via the CI spectrum induced by consistent invertibility property of operators,
we explore property ($UW${\scriptsize \it{E}}) for $T$ and $T^\ast$ simultaneously. Furthermore, the transfer of property ($UW${\scriptsize \it{E}}) from $T$ to $f(T)$ and $f(T^{\ast})$ is obtained, where $f$ is a function which is analytic in a neighborhood of the spectrum of $T$. At last, with the help of the so-called $(A,B)$ entanglement stable spectra, the entanglement stability of property ($UW${\scriptsize \it{E}}) for $2\times 2$ upper triangular operator matrices is investigated.

\end{abstract}

\maketitle
\section{Introduction}

H. Weyl in his celebrated paper \cite{ai} shows that 
$$\bigcap\limits_{K\in K(H)}\sigma(T+K)=\sigma(T)\setminus\{\lambda\in \hbox{iso}\sigma(T): 0<\hbox{dim}N(T-\lambda I)<\infty\},$$
where $\sigma(T)$ and $K(H)$ denote the spectrum of self-adjoint operator $T$ and the set of all compact operators acting on $H$, respectively.
The observation was abstracted into the assertion ``Weyl's theorem holds".
Later, according to Schechter's investigation \cite{cm},
$$\sigma_{w}(T)=\bigcap\limits_{K\in K(H)}\sigma(T+K),$$
where $\sigma_{w}(T)$ denotes the Weyl spectrum of $T$.
Weyl's theorem has thereby evolved into 
$$\sigma(T)\setminus\sigma_{w}(T)=\pi_{00}(T),$$
where $\pi_{00}(T)=\{\lambda\in \hbox{iso}\sigma(T): 0<\hbox{dim} N(T-\lambda I)<\infty\}$.
 Since then, Weyl's theorem has been gradually extended from initial self-adjoint operators to more general operators \cite{ca,cb,cd}. Additionally, there appeared some variants of Weyl's theorem, such as Browder's theorem, a-Weyl's theorem, property ($R$) and so on(see \cite{ce, ba, bb}).  
 As an important part of spectral theory, Weyl type theorem, the general term of Weyl's theorem and its variants, has attracted much attention of scholars.

 In the last two decades, the research on Weyl type theorem has been further enriched. The stability of Weyl type theorem under some perturbations (see\cite{ay, db}), the transfer of Weyl type theorem from operators to functions of operators (see\cite{ay, dc, dd}) and the entanglement stability of Weyl type theorem for operator matrices (see\cite{af, am, de, df, dg}) have drawn many scholars. Also, new Weyl type theorem has appeared one after another.
In this paper, we are interested in a variant of Weyl's theorem, which is called property ($UW${\scriptsize \it{E}}) (see \cite{aw}).

It's known to all that Weyl type theorem is closely related to some local spectral properties.
As a kind of local spectral properties, 
the concept ``consistent in invertibility" originates from the further exploration of Jacobson's Theorem. We say $T\in B(H)$ is consistent in invertibility (abbrev. CI) if, for any $S\in B(H)$, $TS$ admits the same invertibility as $ST$. Gong and Han gave a characterization for such operators (see\cite{ad}).
Djordjević extended the concept from Hilbert spaces to Banach spaces and further studied the consistency in ``weak invertibility”, such as consistent in Fredholm, consistent in Browder, and so on (see\cite{ab}). After that, Cao defined the CI spectrum to describe the conditions for which various forms of “Weyl’s theorem” hold (see\cite{aa}).
Since then, many wonderful studies have been carried out around CI spectrum and Weyl type theorems. Via such a spectrum, Ren, for instance, explored the stability of property $(R)$ under commuting power finite rank perturbations (see\cite{ap}).

Inspired by the property of CI spectrum, the paper devote to exploring how the property ($UW${\scriptsize \it{E}}) survives simultaneously for $T$ and $T^\ast$ in $B(H)$. In addition, the transfer of property ($UW${\scriptsize \it{E}}) from $T$ to $f(T)$ and $f(T^{\ast})$ simultaneously is considered.
Another aim of the paper is to probe the entanglement stability of property ($UW${\scriptsize \it{E}}) for $2\times2$ upper triangular operator matrices with the help of the so-called entanglement stable spectra. 

\section{Preliminaries}

Throughout this paper, ${\mathbb C}$ and ${\Bbb N}$ denote the set of complex numbers and the set of nonnegative integers, respectively.
 Let $H$ and $K$ be infinite dimensional separable complex Hilbert spaces and $B(H,K)$ be the set of bounded linear operators from $H$ to $K$. By convention we write $B(H)$ for $B(H,H)$.
For $T\in B(H)$, let $n(T)$, $d(T)$, $\sigma(T)$ and $\sigma_p(T)$ represent the dimension of the kernel $N(T)$, the codimension of the range $R(T)$, the spectrum and the set of eigenvalues of $T$ respectively.
Let $\Phi_+(H)=\{T\in B(H):$ $n(T)<\infty$ and $R(T)$ is closed$\}$ and $\Phi_-(H)=\{T\in B(H):$ $d(T)<\infty$$\}.$ If $T\in \Phi_+(H)$ (resp. $T\in \Phi_-(H)$ ), $T$ is said to be an upper (resp. lower)  semi-Fredholm operator.
The upper (resp. lower) semi-Fredholm spectrum $\sigma_{SF_+}(T)$ (resp. $\sigma_{SF_-}(T)$) of $T$ is a set of $\lambda\in\mathbb C$ such that $T-\lambda I\notin \Phi_+(H)$ (resp. $T-\lambda I\notin \Phi_-(H)$).
Furthermore, we call $T$ a semi-Fredholm operator
if $T\in \Phi_+(H) \cup\Phi_-(H)$, and the semi-Fredholm spectrum is defined by $\sigma_{SF}(T)=\sigma_{SF_+}(T)\cap\sigma_{SF_-}(T)$.
The index of $T$ is defined as $\hbox{ind}T=n(T)-d(T)$ for $T\in \Phi_+(H) \cup\Phi_-(H)$.
We call $T$ a bounded below operator if $T\in \Phi_+(H)$ and $n(T)=0$.
The approximate point spectrum and the essential approximate point spectrum of $T\in B(H)$ are defined respectively by
$$\sigma_{a}(T)=\{\lambda\in{\Bbb C}: T-\lambda I \hbox{ is not bounded below}\},$$
$$\sigma_{ea}(T)=\{\lambda\in{\Bbb C}: T-\lambda I\notin\Phi_+(H)  \hbox{ or ind}(T-\lambda I)>0\}.$$
For the sake of simplicity,  we write $\sigma_{ SF_{-}^{+}}(T)=\{\lambda\in{\Bbb C}: T-\lambda I\notin\Phi_-(H)  \hbox{ or }\\\hbox{ind}(T-\lambda I)<0\}.$
$T$ is called a Fredholm operator if $T\in\Phi_+(H) \cap\Phi_-(H)$.
If $T$ is Fredholm with $\hbox{ind} T=0$, we call $T$ a Weyl operator, and the Weyl spectrum $\sigma_w(T)$ of $T$ is defined as the set of $\lambda\in\mathbb C$ such that $T-\lambda I$ is not a Weyl operator.

The ascent (resp. descent) of $T$ is defined as $\hbox{asc}(T)=\inf\{n\in{\Bbb N}: N(T^{n})=N(T^{n+1})\}$ (resp. $\hbox{des}(T)=\inf\{n\in{\Bbb N}: R(T^{n})=R(T^{n+1})\}$).
If the infimum does not exist, we write $\hbox{asc}(T)=\infty$ (resp. $\hbox{des}(T)=\infty$).
If $T\in \Phi_+(H)$ with $\hbox{asc}(T)<\infty$, we call $T$ an upper semi-Browder operator.
$T$ is said to be a Browder operator if it is Fredholm of finite ascent and descent.
The upper semi-Browder spectrum $\sigma_{ab}(T)$ (resp. the Browder spectrum $\sigma_b(T)$) is the set of $\lambda\in\mathbb C$ such that $T-\lambda I$ is not an upper semi-Browder operator (resp. a Browder operator).

Set $\rho(T)={\Bbb C}\backslash\sigma(T)$ and $\rho_{\ast}(T)={\Bbb C}\backslash\sigma_{\ast}(T)$,
where  $\sigma_{\ast}\in\{ \sigma_{SF_+},  \sigma_{SF_-}, \sigma_{SF},\\ \sigma_{a}, \sigma_{ea}, \sigma_{ SF_{-}^{+}}, \sigma_{w}, \sigma_{ab}, \sigma_{b}\}$.
Also, write $\Pi(T)=\sigma(T)\backslash\sigma_p(T)$ and $\Pi_{x}(T)=\sigma_{x}(T)\backslash\sigma_p(T)$, where
$\sigma_{x}\in\{  \sigma_{a}, \sigma_{G}\}$ and $\sigma_G(T)=\{\lambda\in\Bbb C: R(T-\lambda I)$ is not closed$\}$.

For a subset $E \subseteq{\Bbb C}$, we denote by $\hbox{iso}E$, $\partial E$ and $\hbox{acc} E$ the set of isolated points, boundary points and accumulative points of $E$ respectively.
Then denote $E(T)=\hbox{iso}\sigma(T)\cap\sigma_p(T)$, namely, the set of the whole isolated eigenvalues of $T$.

Let's recall the concepts used throughout the paper in turn.

\begin{defn}\cite[Definition 3.1]{aw} 
An operator $T\in B(H)$ is said to satisfy property ($UW${\scriptsize \it{E}}), denoted by $T\in$($UW${\scriptsize \it{E}}), if$$\sigma_{a}(T)\setminus\sigma_{ea}(T)=E(T).$$
\end{defn}

 One can see that property ($UW${\scriptsize \it{E}}) can
reason out Weyl's theorem, and furthermore a-Weyl's theorem; whereas, as mentioned in \cite{af}, property ($UW${\scriptsize \it{E}}) and a-Weyl's theorem cannot reason from each other.
For $T\in B(H)$,  $T^{\ast}$ does not necessarily possess ($UW${\scriptsize \it{E}}) provided that $T\in$($UW${\scriptsize \it{E}}). Indeed, set $T_1\in B(\ell^2)$ be defined by:
$T_1(x_1, x_2, x_3, \cdots)=(0,x_1,\frac{x_2}{2}, \frac{x_3}{3}, \cdots)$. Note that $n(T_1)=0$ whereas $n({T_1}^{\ast})=1$, which shows that $T_1\in$($UW${\scriptsize \it{E}}) but ${T_1}^{\ast}\notin$($UW${\scriptsize \it{E}}). By means of the CI spectrum given below, this paper will give a necessary and sufficient condition that property ($UW${\scriptsize \it{E}}) holds for $T$ and $T^{\ast}$ simultaneously.

\begin{defn}\cite{ad} 
An operator $T\in B(H)$ is said to be consistent in invertibility (abbrev. CI), if $ST$ and $TS$ are either both or neither invertible for any $S\in B(H)$.
\end{defn}

A conclusion critical to CI operators was given in \cite{ad} as follows.

\begin{lem}
\cite[Theorem 1.1]{ad} Let $T\in B(H)$. Then $T$ is a CI operator if and only if one of the following three mutually disjoint cases occurs:

 $1)$\  $T$ is invertible;
 
$2)$\  $R(T)$ is not closed;

$3)$\  $N(T)\neq\{0\}$ and $R(T)=\overline{R(T)}\neq H$.
\end{lem}

By the set 
$$\sigma_{CI}(T)=\{\lambda\in{\Bbb C}: T-\lambda I\hbox{ is not CI}\},$$ 
the CI spectrum of $T$ is denoted. One can see that $\lambda\in \sigma_{CI}(T)$ if and only if $T-\lambda I$ is either bounded below but not invertible or surjective but not invertible. Obviously, $\sigma_{CI}(T)$ is an open subset of the spectrum of $T$ (see \cite{aa}).
For more conclusions concerning CI, one can refer to \cite{ae, ac, az}.

For $T\in B(H)$, if $K$ is a clopen subset of $\sigma(T)$, there is an analytic Cauchy domain $\Omega$ such that $K\subseteq\Omega$ and $[\sigma(T)\setminus K]\cap\overline{\Omega}=\emptyset$, where $\overline{\Omega}$ is the closure of $\Omega$. Let $E(K;T)$ denote the associated spectral idempotent 
of $T$ corresponding to $K$, i.e.,
$$E(K;T)=\frac{1}{2\pi \hbox{i}}\int\limits_{\Gamma}(\lambda I-T)^{-1}d\lambda,$$
where $\Gamma=\partial\Omega$ is positively oriented with respect to $\Omega$.
Denote $H(K;T)=R(E(K;T))$, and write $H(\lambda;T)$ instead of $H(\{\lambda\};T)$ provided that $\lambda\in\hbox{iso}\sigma(T)$. It follows from
\cite[Theorem 3.6]{bd} that $\lambda$ is a Riesz point of $T$, denoted by $\lambda\in\sigma_{0}(T)$, when $\dim H(\lambda;T)<\infty$. Thereinto, $\sigma_{0}(T)$ is the set of Riesz points of $T$ which are in $\sigma(T)$. From \cite[Proposition 2]{bc}, one can see $\sigma_{0}(T)=\sigma(T)\setminus\sigma_{b}(T)$.

\section{ Property ($UW_E$) for operator and its conjugate}

This section aims to describe property ($UW${\scriptsize \it{E}}) for $T\in B(H)$ and its conjugate $T^*$ simultaneously in the light of CI spectrum.
Let's begin with the discussion of property ($UW${\scriptsize \it{E}}) for $T\in B(H)$.

\begin{prop}
$T\in$($UW${\scriptsize \it{E}}) if and only if
 $\sigma_{b}(T)=\sigma_{CI}(T)\cup \Pi_G(T)\cup[\hbox{acc}\sigma(T)\cap\sigma_{ea}(T)]$.
\end{prop}

\begin{proof}
For the sufficiency, since the set $\big\{ [\sigma_{a}(T)\setminus\sigma_{ea}(T)]\cup E(T)\big\}\cap\big \{\sigma_{CI}(T)\cup \Pi_G(T)\cup[\hbox{acc}\sigma(T)\cap\sigma_{ea}(T)]\big\}=\emptyset$, we have $[\sigma_{a}(T)\setminus\sigma_{ea}(T)]\cup E(T)\subseteq\sigma_{0}(T)$. Then
$\sigma_{a}(T)\setminus\sigma_{ea}(T)=E(T)$, and $T\in$($UW${\scriptsize \it{E}}).

 In contrast, let $\lambda_0\notin\sigma_{CI}(T)\cup \Pi_G(T)\cup[\hbox{acc}\sigma(T)\cap\sigma_{ea}(T)]$,
 then $\lambda_0\in\rho_{CI}(T)$ and we have either $R(T-\lambda_0 I)$ is closed or $n(T-\lambda_0 I)>0$.
 In view of Lemma 2.3, there are only two cases should be considered. One is that $T-\lambda_0 I$ is invertible, the other is $n(T-\lambda_0 I)>0$. Without loss of generality, we assume that $n(T-\lambda_0 I)>0$. Note that $\lambda_0\notin\hbox{acc}\sigma(T)\cap\sigma_{ea}(T)$. So either $\lambda_0\in E(T)$ or $\lambda_0\in\sigma_{a}(T)\setminus\sigma_{ea}(T)$. It follows from $T\in$($UW${\scriptsize \it{E}}) that $\lambda_0\notin\sigma_b(T)$.
 The converse inclusion is clear.
\end{proof}

 Similarly, we can show that $T\in$($UW${\scriptsize \it{E}}) if and only if
 $\sigma_{b}(T)=\sigma_{CI}(T)\cup \Pi_a(T)\cup[\hbox{acc}\sigma(T)\cap\sigma_{ea}(T)]$.

 Considering the fact that $\sigma_{CI}(T)$ is open, it is natural to consider whether property ($UW${\scriptsize \it{E}}) can be described by $\overline{\sigma_{CI}(T)}$. The answer is positive with the help of Kato spectrum.
 Here we say $T\in B(H)$ is Kato if $N(T)\subseteq\bigcap\limits_{n=1}^\infty R(T^n)$, and call 
$\sigma_K(T)=\{\lambda\in{\Bbb C}: T-\lambda I$ is not Kato$\}$ the Kato spectrum of $T$.

\begin{cor} The following statements are equivalent.

$1)$\ $T\in$($UW${\scriptsize \it{E}});

$2)$\ $\sigma_{b}(T)=[\overline{\sigma_{CI}(T)}\cap\sigma_{ea}(T)]\cup \Pi(T)\cup\hbox{acc}\sigma_{ea}(T)\cup\hbox{acc}\sigma_K(T)$;

$3)$\ $\sigma_{b}(T)=[\overline{\sigma_{CI}(T)}\cap\sigma_{G}(T)]\cup \Pi(T)\cup\hbox{acc}\sigma_{ea}(T)\cup\hbox{acc}\sigma_K(T)\cup\{\lambda\in\hbox{acc}\sigma(T): n(T-\lambda I)=\infty\}$.

\end{cor}

\begin{proof} $1)\Rightarrow 2)$.\  Since $T\in$($UW${\scriptsize \it{E}}), we can get that $\hbox{iso}\sigma_{ea}(T)\subseteq\hbox{iso}\sigma_{a}(T)\cup\hbox{acc}\sigma_{K}(T)$ and  $\sigma_{b}(T)=\sigma_{CI}(T)\cup \Pi_a(T)\cup[\hbox{acc}\sigma(T)\cap\sigma_{ea}(T)]$.
Thus
$$\begin{array}{rcl}
\hbox{acc}\sigma(T)\cap\sigma_{ea}(T)&=&[\hbox{acc}\sigma(T)\cap\hbox{acc}\sigma_{ea}(T)]\cup[\hbox{acc}\sigma(T)\cap\hbox{iso}\sigma_{ea}(T)]\\
&\subseteq&\hbox{acc}\sigma_{ea}(T)\cup\hbox{acc}\sigma_{K}(T)\cup[\hbox{acc}\sigma(T)\cap\hbox{iso}\sigma_{a}(T)]\\
&\subseteq&\hbox{acc}\sigma_{ea}(T)\cup\hbox{acc}\sigma_{K}(T)\cup[\overline{\sigma_{CI}(T)}\cap\sigma_{ea}(T)].
\end{array}$$
Additionally, in view of $\sigma_{CI}(T)=[\overline{\sigma_{CI}(T)}\cap\sigma_{ea}(T)]\cup \Pi(T)$, we have $$\sigma_{b}(T)\subseteq[\overline{\sigma_{CI}(T)}\cap\sigma_{ea}(T)]\cup \Pi(T)\cup\hbox{acc}\sigma_{ea}(T)\cup\hbox{acc}\sigma_K(T).$$
The converse is obvious.

 $2)\Rightarrow 3)$.\ Since $\sigma_{ea}(T)\subseteq\sigma_{G}(T)\cup\{\lambda\in{\Bbb C}: n(T-\lambda I)=\infty\}\cup\hbox{acc}\sigma_{ea}(T)$, one has $\overline{\sigma_{CI}(T)}\cap\sigma_{ea}(T)\subseteq[\overline{\sigma_{CI}(T)}\cap\sigma_{G}(T)]\cup\hbox{acc}\sigma_{ea}(T)\cup\{\lambda\in\hbox{acc}\sigma(T): n(T-\lambda I)=\infty\}$, and the desired result holds.

 $3)\Rightarrow 1)$.\  By the definition of property ($UW${\scriptsize \it{E}}), the proof is clear.
\end{proof}

In order to explore property ($UW${\scriptsize \it{E}}) for $T$ and $T^{\ast}$ simultaneously, let's characterize the necessary and suﬀicient condition for $T^{\ast}\in$($UW${\scriptsize \it{E}}) first with the help of the CI spectrum.

\begin{lem}  Let $T\in B(H)$. Then $T^{\ast}\in$($UW${\scriptsize \it{E}}) if and only if
 $\sigma_{b}(T)=\sigma_{CI}(T)\cup\{\lambda\in\sigma_G(T): n(T^{\ast}-\overline{\lambda} I)=0\}\cup[\hbox{acc}\sigma(T)\cap\sigma_{SF_{-}^{+}}(T)]$.
\end{lem}

 \begin{proof}For the necessity, assume that $\lambda_0\notin\sigma_{CI}(T)\cup\{\lambda\in\sigma_G(T): n(T^{\ast}-\overline{\lambda} I)=0\}\cup[\hbox{acc}\sigma(T)\cap\sigma_{SF_{-}^{+}}(T)]$. We can deduce $\overline{\lambda_0}\notin\sigma_{CI}(T^{\ast})\cup \Pi_G(T^{\ast})\cup[\hbox{acc}\sigma(T^{\ast})\cap\sigma_{ea}(T^{\ast})]$. In virtue of $T^{\ast}\in$($UW${\scriptsize \it{E}}), it follows from Proposition 3.1 that $T^{\ast}-\overline{\lambda_0} I$ is Browder. Thus $\lambda_0\notin\sigma_b(T)$.

On the contrary, if $\overline{\lambda_0}\in\sigma_{a}(T^{\ast})\setminus\sigma_{ea}(T^{\ast})$, $\lambda_0\notin\sigma_{CI}(T)\cup\{\lambda\in\sigma_G(T): n(T^{\ast}-\overline{\lambda} I)=0\}\cup[\hbox{acc}\sigma(T)\cap\sigma_{SF_{-}^{+}}(T)]$, and furthermore, doesn't belong to $\sigma_{b}(T)$. Hence $\overline{\lambda_0}\in E(T^{\ast})$. The converse is equally true.
 \end{proof}

\begin{lem} Suppose $T\in$($UW${\scriptsize\it{E}}). Then $T^{\ast}\in$($UW${\scriptsize \it{E}}) if and only if $\sigma_b(T)=\sigma_{CI}(T)\cup \Pi^{\ast}(T)\cup[\hbox{acc}\sigma(T)\cap\sigma_{SF_{-}^{+}}(T)]$, where $\Pi^{\ast}(T)=\{\lambda\in\sigma(T): n(T-\lambda I)=n(T^{\ast}-\overline{\lambda} I)=0\}$ . 
\end{lem}

 \begin{proof}From Lemma 3.3, it suffices to show the necessity. Let $\lambda_0\notin\sigma_{CI}(T)\cup \Pi^{\ast}(T)\cup[\hbox{acc}\sigma(T)\cap\sigma_{SF_{-}^{+}}(T)]$. Without loss of generality, we assume $\lambda_0\in\sigma(T)$. It follows from $\lambda_0\notin\sigma_{CI}(T)$ that either $R(T-\lambda_0 I)$ is not closed or  $R(T-\lambda_0 I)$ is closed with $n(T-\lambda_0 I)\cdot d(T-\lambda_0 I)>0$.
 
  Suppose firstly that $R(T-\lambda_0 I)$ is not closed. One can see $\lambda_0\in\hbox{iso}\sigma(T)$. Note that $\lambda_0\notin \Pi^{\ast}(T)$, then we have either $n(T-\lambda_0 I)\neq 0$ or $n(T-\lambda_0 I)=0$ whereas $n(T^{\ast}-\overline{\lambda_0} I)\neq0$. Moreover, it deduces from $T\in$($UW${\scriptsize \it{E}}) that $\lambda_0\notin\sigma_b(T)$ provided that $n(T-\lambda_0 I)\neq 0$. If $n(T-\lambda_0 I)=0$ and $n(T^{\ast}-\overline{\lambda_0} I)\neq0$, in view of $T^{\ast}\in$($UW${\scriptsize \it{E}}), we can also get that $\lambda_0\notin\sigma_b(T)$.  
  
  Suppose secondly that $R(T-\lambda_0 I)$ is closed with $n(T-\lambda_0 I)\cdot d(T-\lambda_0 I)>0$.
We have either $\lambda_0\in\hbox{iso}\sigma(T)$ or $\lambda_0\notin\sigma_{SF_{-}^{+}}(T)$. If $\lambda_0\in\hbox{iso}\sigma(T)$, $\lambda_0\notin\sigma_b(T)$ in terms of either $T\in$($UW${\scriptsize \it{E}}) or $T^{\ast}\in$($UW${\scriptsize \it{E}}). Provided $\lambda_0\notin\sigma_{SF_{-}^{+}}(T)$, we can see that $\overline{\lambda_0}\in\sigma_a(T^{\ast})\setminus\sigma_{ea}(T^{\ast})$. 
Using $T^{\ast}\in$($UW${\scriptsize \it{E}}) once again, we can finally deduce that $T-\lambda_0 I$ is Browder.

The converse is obvious. Thus a desired result emerges.  
\end{proof}

\begin{thm} Let $T\in B(H)$. Then both $T$ and $T^{\ast}$ are in ($UW${\scriptsize \it{E}}) if and only if
 $\sigma_{b}(T)=\sigma_{CI}(T)\cup \Pi^{\ast}(T)\cup[\hbox{acc}\sigma(T)\cap\sigma_{SF}(T)]$.
\end{thm}
\begin{proof} If $T$ and $T^{\ast}$ are both in ($UW${\scriptsize \it{E}}), $\sigma_b(T)=\Pi^{\ast}(T)\cup\hbox{acc}\sigma(T)$. Now we claim that $\sigma_b(T)=\sigma_{CI}(T)\cup\sigma_{SF}(T)$.

Indeed, suppose $\lambda_0\notin\sigma_{CI}(T)\cup\sigma_{SF}(T)$, then $R(T-\lambda_0 I)$ is closed, which infers that either $\lambda_0\in\rho(T)$ or $n(T-\lambda_0 I)\cdot d(T-\lambda_0 I)>0$. Without loss of generality we can assume that $n(T-\lambda_0 I)\cdot d(T-\lambda_0 I)>0$. In view of $\lambda_0\in\rho_{SF}(T)$, $\lambda_0\in\sigma_{a}(T)\setminus\sigma_{ea}(T)$ or $\overline{\lambda_0}\in\sigma_{a}(T^{\ast})\setminus\sigma_{ea}(T^{\ast})$. Since $T\in$($UW${\scriptsize \it{E}}) and $T^{\ast}\in$($UW${\scriptsize \it{E}}), we find that $\lambda_0\notin\sigma_{b}(T)$. The converse is apparent. Hence $\sigma_{b}(T)=\sigma_{CI}(T)\cup\Pi^{\ast}(T)\cup[\hbox{acc}\sigma(T)\cap\sigma_{SF}(T)]$.

 Conversely, the condition shows that $\sigma_{b}(T)\subseteq\sigma_{CI}(T)\cup\Pi_G(T)\cup[\hbox{acc}\sigma(T)\cap\sigma_{ea}(T)]$. In view of Proposition 3.1,  $T\in$($UW${\scriptsize \it{E}}). And furthermore, note that $\sigma_{b}(T)=\sigma_{CI}(T)\cup \Pi^{\ast}(T)\cup[\hbox{acc}\sigma(T)\cap\sigma_{SF}(T)]\subseteq\sigma_{CI}(T)\cup\Pi^{\ast}(T)\cup[\hbox{acc}\sigma(T)\cap\sigma_{SF_{-}^{+}}(T)]$. According to Lemma 3.4, we can deduce that $T^{\ast}\in$($UW${\scriptsize \it{E}}).
  \end{proof}

\begin{rem}
 If both $T$ and $T^{\ast}$ are in ($UW${\scriptsize \it{E}}), from the following examples one can see three parts which form $\sigma_{b}(T)$ together are essential in Theorem 3.5.

$(i)$\ Let $T_2\in B(\ell^2)$ be defined by:
$T_2(x_1, x_2, x_3, \cdots)=(0, x_{1}, x_{2}, \cdots)$. Then $T_2\in$($UW${\scriptsize \it{E}}) and ${T_2}^{\ast}\in$($UW${\scriptsize \it{E}}).
But $\Pi^{\ast}(T_2)\cup[\hbox{acc}\sigma(T_2)\cap\sigma_{SF}(T_2)]=\{\lambda\in{\Bbb C}: \lvert \lambda \rvert=1\}\neq\sigma_{b}(T_2)$. Hence ``$\sigma_{CI}(T_2)$" cannot be dropped.

$(ii)$\ Let $A, B,C\in B(\ell^2)$ be defined by:
$A(x_1, x_2, x_3, \cdots)=(0, x_1,\frac{x_2}{2}, \frac{x_3}{3},\\ \cdots)$, $B(x_1, x_2, x_3, \cdots)=(x_2, \frac{x_3}{2}, \frac{x_4}{3}, \cdots)$, and $C(x_1, x_2, x_3, \cdots)=(x_1, 0,0, \cdots)$.
Set $T_3=\bigl(\begin{smallmatrix} A&C\\0&B \end{smallmatrix}\bigr)$. Then $\sigma(T_3)=\{0\}$ and $n(T_3)=n({T_3}^{\ast})=0$, and thus both $T_3$ and ${T_3}^{\ast}$ are in ($UW${\scriptsize \it{E}}). However, $\sigma_{CI}(T_3)\cup[\hbox{acc}\sigma(T_3)\cap\sigma_{SF}(T_3)]=\emptyset\neq\sigma_{b}(T_3)$. This implies that the demand ``$\Pi^{\ast}(T_3)$" cannot be dropped.

$(iii)$\ Let $A, B\in B(\ell^2)$ be defined by:
$A(x_1, x_2, x_3, \cdots)=(0, x_{1}, 0, x_{2}, \cdots)$, $B(x_1, x_2, x_3, \cdots)=({x_2}, {x_4}, \cdots)$. Put $T_4=\bigl(\begin{smallmatrix} A&0\\0&B \end{smallmatrix}\bigr)$. Then $T_4\in$($UW${\scriptsize \it{E}}) and ${T_4}^{\ast}\in$($UW${\scriptsize \it{E}}). But $\sigma_{CI}(T_4)\cup\Pi^{\ast}(T_4)=\{\lambda\in{\Bbb C}: \lvert \lambda \rvert=1\}\neq\sigma_{b}(T_4)$. Therefore, the demand ``$\hbox{acc}\sigma(T_4)\cap\sigma_{SF}(T_4)$" cannot be dropped.
 \end{rem}

\section{Transfer of property ($UW_E$)}

Let Hol($\sigma(T)$) denote the set of functions $f$ which are analytic on some neighborhood of $\sigma(T)$. For $f\in$Hol$(\sigma(T))$, $f(T)$ denotes the holomorphic functional calculus of $T$ with respect to $f$.
Notice that one can not deduce $f(T)\in$($UW${\scriptsize \it{E}}) for $f\in$Hol$(\sigma(T))$ provided that $T\in$($UW${\scriptsize \it{E}}). 
For instance, put $A, B\in B(\ell^2)$ be defined by:
$A(x_1, x_2, x_3, \cdots)=(0, x_{1}, x_{2}, \cdots)$, $B(x_1, x_2, x_3, \cdots)=(x_{2}, x_{3},x_4, \cdots)$.
Set $T_5=\bigl(\begin{smallmatrix} A+I&0\\0&B-I \end{smallmatrix}\bigr)$. It is clear that $T_5\in$($UW${\scriptsize \it{E}}). However, $(T_5-I)(T_5+I)$ is Weyl but not Browder, thus it is not in ($UW${\scriptsize \it{E}}).

The aim of this section is to probe the transfer of property ($UW${\scriptsize \it{E}}) from $T$ to $f(T)$.  Moreover, the transfer of property ($UW${\scriptsize \it{E}}) from $T$ to $f(T)$ and $f(T^{\ast})$ simultaneously is explored.
Now let's explore property ($UW${\scriptsize \it{E}}) for $f(T)$ first.

\begin{lem} Let $T\in$($UW${\scriptsize \it{E}}). Then $f(T)\in$($UW${\scriptsize \it{E}}) for any $f\in$Hol$(\sigma(T))$ if and only if one of the two following conditions holds:

$1)$\ $\hbox{ind}(T-\lambda I)\cdot \hbox{ind}(T-\mu I)\geq 0$ for any $\lambda,\mu\in\rho_{{SF}_{+}}(T)$ whenever $\sigma_0(T)=\emptyset$;

$2)$\  $\sigma_b(T)=[\overline{\sigma_{CI}(T)}\cap\sigma_{ea}(T)]\cup\hbox{acc}[\rho_{CI}(T)\cap\sigma_{ea}(T)]\cup\hbox{acc}\sigma_K(T)$.
 \end{lem}

\begin{proof} Suppose $f(T)\in$($UW${\scriptsize \it{E}}) for any $f\in$Hol$(\sigma(T))$. Let $\sigma_0(T)=\emptyset$. Suppose on the contrary that there exist $\lambda_0,\mu_0\in\rho_{SF_+}(T)$ with $\hbox{ind}(T-\lambda_0I)<0$ and $\hbox{ind}(T-\mu_0I)=n$, where $n$ is a positive integer. Let $f_1(T)=(T-\lambda_0 I)^{n}(T-\mu_0I)^{m}$ when $\hbox{ind}(T-\lambda_0I)=-m$, where $m$ is a positive integer, or $f_1(T)=(T-\lambda_0 I)(T-\mu_0I)$ when $\hbox{ind}(T-\lambda_0I)=-\infty$. Then $0\in\sigma_a(f_1(T))\backslash\sigma_{ea}(f_1(T))$ in either of the two cases. In view of $f_1(T)\in$($UW${\scriptsize \it{E}}), $f_1(T)$ is Browder, so is $T-\lambda_0I$, which contradicts with the fact ``$\hbox{ind}(T-\lambda_0I)<0$".

If $\sigma_0(T)\neq\emptyset$, let $\lambda_1\in\sigma_0(T)$.  We claim that $\sigma_{ea}(T)=\sigma_b(T)$ and $\hbox{iso}\sigma(T)=\sigma_0(T)$.
Otherwise, choose $\lambda_2\in\sigma_b(T)\backslash\sigma_{ea}(T)$. Then put $f_2(T)=(T-\lambda_1 I)(T-\lambda_2I)$. One has $T-\lambda_2I$ is Browder, a contradiction.
For the other equality, suppose $\lambda_3\in\hbox{iso}\sigma(T)\backslash\sigma_0(T)$.
Put $K_{\alpha}=\{\lambda_{1}\}$, $K_{\beta}=\{\lambda_3\}$ and $K_{\gamma}=\sigma(T)\setminus(K_{\alpha}\cup K_{\beta})$.
Using \cite[Theorem 2.10]{ag}, $T$ can be written as $$T=\left(\begin{array}{cccc}T_{\alpha}&0&0\\0&T_{\beta}&0\\0&0&T_{\gamma}\\\end{array}\right)
\begin{array}{cccc}H(K_{\alpha};T)\\H(K_{\beta};T)\\H(K_{\gamma};T)\\\end{array}{\hbox{,}}$$
where $\sigma(T_i)=K_i$, $i=\alpha,\beta,\gamma.$
Put $f_3(T)=(T-\lambda_1I)(T-\lambda_3I)$.
Then $$f_3(T)=\left(\begin{array}{cccc}f_3(T_\alpha)&0&0\\0&f_3(T_\beta)&0\\0&0&f_3(T_\gamma)
\\\end{array}\right)\begin{array}{cccc}H(K_{\alpha};T)\\H(K_{\beta};T)\\H(K_{\gamma};T)
\\\end{array}{\hbox{.}}$$
It follows from the spectral mapping theorem that $\sigma(f_3(T_{i}))=f_3(\sigma(T_{i}))=\{0\},$ where $i=\alpha,\beta.$ Note that $0\notin\sigma(f_3(T_\gamma))$. Then $0\in\hbox{iso}\sigma(f_3(T))$. Besides,  by $n(f_3(T))\geq n(T-\lambda_1I)>0$, one has $0\in E(f_3(T))$, and $f_3(T)\in$($UW${\scriptsize \it{E}}). Thus $f_3(T)$, and furthermore, $T-\lambda_2 I$ is Browder, which is a contradiction. Therefore, the above assertion holds.

Next we show that $\sigma_{b}(T)$ has the desired decomposition. Take arbitrarily $\lambda\notin[\overline{\sigma_{CI}(T)}\cap\sigma_{ea}(T)]\cup\hbox{acc}[\rho_{CI}(T)\cap\sigma_{ea}(T)]\cup\hbox{acc}\sigma_K(T)$. Without loss of generality, we may assume that $\lambda\notin\overline{\sigma_{CI}(T)}$. Then there is a $\varepsilon>0$ such that $\mu\notin\sigma_{ea}(T)$ for any $\mu$ with $0<\rvert \mu-\lambda\lvert<\varepsilon$. Thus $T-\mu I$ is Browder. Furthermore, it follows from $\lambda\notin\hbox{acc}\sigma_K(T)$ \cite[Lemma 3.4]{ah} that  $T-\mu I$ is invertible. Therefore  $\lambda\in\hbox{iso}\sigma(T)\cup\rho(T)\subseteq\rho_b(T)$.

Now we explore the sufficiency according to the case $\sigma_0(T)=\emptyset$ or not.

Suppose firstly that $\sigma_0(T)=\emptyset$. It shows from $T\in$($UW${\scriptsize \it{E}}) that $\sigma_{a}(T)=\sigma_{ea}(T)$ and $E(T)=\emptyset$. Since $\hbox{ind}(T-\lambda I)\cdot \hbox{ind}(T-\mu I)\geq 0$ for any $\lambda,\mu\in\rho_{{SF}_{+}}(T)$, we see that $\sigma_{ea}(f(T))=f(\sigma_{ea}(T))$ for any $f\in Hol(\sigma(T))$. Then $\sigma_{a}(f(T))=f(\sigma_{a}(T))=f(\sigma_{ea}(T))=\sigma_{ea}(f(T))$. In view of $E(f(T))\subseteq f(E(T))$, $E(f(T))=\emptyset$. So $f(T)\in$($UW${\scriptsize \it{E}}).

Suppose secondly that $\sigma_0(T)\neq\emptyset$.
It follows from the decomposition of $\sigma_{b}(T)$ that $\hbox{ind}(T-\lambda I)\geq 0$ for any $\lambda\in\rho_{SF_+}(T)$, $\sigma_{w}(T)=\sigma_b(T)$ and $\hbox{iso}\sigma(T)=\sigma_0(T)$.  For any $f\in Hol(\sigma(T))$, choose arbitrarily $\mu_0\in\sigma_{a}(f(T))\backslash\sigma_{ea}(f(T))$.
Put $$f(T)-\mu_0 I=a(T-\lambda_{1}I)^{n_{1}}(T-\lambda_{2}I)^{n_{2}}\cdots(T-\lambda_{t}I)^{n_{t}}g(T),$$ where $\lambda_{i}\neq \lambda_{j}$ if $i\neq j$ and $g(T)$ is invertible. Since $\hbox{ind}(T-\lambda I)\geq 0$ for any $\lambda\in\rho_{SF_+}(T)$, $T-\lambda_i I$ is Weyl, and furthermore, Browder for $1\leq i\leq t$. Thus $f(T)-\mu_0 I$ is Browder.
Finally, we check $E(f(T))\subseteq \sigma_{a}(f(T))\backslash\sigma_{ea}(f(T))$.
Choose arbitrarily $\mu_0\in E(f(T))$. Again put $$f(T)-\mu_0 I=a(T-\lambda_{1}I)^{n_{1}}(T-\lambda_{2}I)^{n_{2}}\cdots(T-\lambda_{t}I)^{n_{k}}r(T),$$ where $\lambda_{i}\neq \lambda_{j}$ if $i\neq j$ and $r(T)$ is invertible. Without loss of generality we assume that $\lambda_i\in \sigma(T)$, $1\leq i\leq t$. Then $\lambda_i\in \hbox{iso}\sigma(T)=\sigma_0(T)$. Therefore, $f(T)-\mu_0 I$ is Browder. The proof is completed.
\end{proof}

In order to explore the transfer of property ($UW${\scriptsize \it{E}}) from $T$ to $f(T)$, it is necessary to give the following proposition as a preparation.

 \begin{prop} Let $T\in B(H)$. Then $f(T)\in$($UW${\scriptsize \it{E}}) for any $f\in$Hol$(\sigma(T))$ if and only if one of the following three cases occurs:

$1)$\ $\sigma(T)=[\overline{\sigma_{CI}(T)}\cap\sigma_{SF}(T)]\cup\Pi(T)\cup\hbox{acc}\sigma_{SF_+}(T)$;

$2)$\ $\sigma(T)=[\overline{\sigma_{CI}(T)}\cap\sigma_{ea}(T)]\cup\hbox{acc}[\rho_{CI}(T)\cap\sigma_{ea}(T)]\cup\Pi_a(T)$;

$3)$\ $\sigma_b(T)=[\overline{\sigma_{CI}(T)}\cap\sigma_{ea}(T)]\cup\hbox{acc}[\rho_{CI}(T)\cap\sigma_{ea}(T)]\cup\hbox{acc}\sigma_K(T)$.
 \end{prop}

\begin{proof}
If $\sigma_0(T)\neq\emptyset$, then $f(T)\in$($UW${\scriptsize \it{E}}) for any $f\in$Hol$(\sigma(T))$ if and only if $\sigma_b(T)=[\overline{\sigma_{CI}(T)}\cap\sigma_{ea}(T)]\cup\hbox{acc}[\rho_{CI}(T)\cap\sigma_{ea}(T)]\cup\hbox{acc}\sigma_K(T)$.
If $\sigma_0(T)=\emptyset$ and $f(T)\in$($UW${\scriptsize \it{E}}) for any $f\in$Hol$(\sigma(T))$, then $\sigma_{a}(T)=\sigma_{ea}(T)$ and $E(T)=\emptyset$, and meanwhile,
we have $\hbox{ind}(T-\lambda I)\cdot \hbox{ind}(T-\mu I)\geq 0$ for any $\lambda,\mu\in\rho_{{SF}_{+}}(T)$.

Therefore, it suffices to check the following two assertions.

 Claim I: The validity of the relation $\hbox{ind}(T-\lambda I)\leq 0$ for any $\lambda\in\rho_{SF_+}(T)$ is equivalent to the equality $\sigma(T)=[\overline{\sigma_{CI}(T)}\cap\sigma_{SF}(T)]\cup\Pi(T)\cup\hbox{acc}\sigma_{SF_+}(T)$.

 Indeed, assume on the contrary that there is some $\lambda_0\in\rho_{SF_+}(T)$ such that $\hbox{ind}(T-\lambda_0 I)> 0$,  one can see $\lambda_0\notin[\overline{\sigma_{CI}(T)}\cap\sigma_{SF}(T)]\cup\Pi(T)\cup\hbox{acc}\sigma_{SF_+}(T)$,
 and furthermore, doesn't belong to $\sigma(T)$, which is a contradiction.
In contrast, choose arbitrarily $\lambda_0\notin[\overline{\sigma_{CI}(T)}\cap\sigma_{SF}(T)]\cup\Pi(T)\cup\hbox{acc}\sigma_{SF_+}(T)$.
Now we prove that $\lambda_0\notin\sigma(T)$. Suppose otherwise that $\lambda_0\in\sigma(T)$, we can see $n(T-\lambda_0 I)> 0$, that is, $\lambda_0\in\sigma_a(T)$. Then we have the following two cases.
\begin{enumerate}
  \item $\lambda_0\notin\sigma_{SF}(T)$. It follows from the perturbation theory of semi-Fredholm operators \cite[XI. 3.11]{fm} that $\lambda_0\in\sigma_{a}(T)\backslash\sigma_{ea}(T)$, which contradicts with ``$\sigma_{a}(T)=\sigma_{ea}(T)$".
  \item $\lambda_0\notin\overline{\sigma_{CI}(T)}$. Then $\lambda\in\rho_{CI}(T)\cap\rho_{SF_+}(T)=\rho(T)$ if $0<\lvert \lambda-\lambda_0\rvert$ small enough. Then $\lambda_0\in\hbox{iso}\sigma(T)$. It deduces from $n(T-\lambda_0 I)> 0$ that $\lambda_0\in E(T)$, which contradicts with
`` $E(T)=\emptyset$".
\end{enumerate}
Therefore we conclude that $\lambda_0\notin\sigma(T)$. The claim I is proved.

 Claim II: The validity of the relation $\hbox{ind}(T-\lambda I)\geq 0$ for any $\lambda\in\rho_{SF_+}(T)$ is equivalent to the equality $\sigma(T)=[\overline{\sigma_{CI}(T)}\cap\sigma_{ea}(T)]\cup\hbox{acc}[\rho_{CI}(T)\cap\sigma_{ea}(T)]\cup\Pi_a(T)$.
 
 Indeed, analogous to the Claim I above, the sufficiency is clear. For the necessity, choose arbitrarily $\lambda_0\notin[\overline{\sigma_{CI}(T)}\cap\sigma_{ea}(T)]\cup\hbox{acc}[\rho_{CI}(T)\cap\sigma_{ea}(T)]\cup\Pi_a(T)$,
then $\lambda_0\notin\sigma(T)$.
Assume otherwise that $\lambda_0\in\sigma(T)$, then $n(T-\lambda_0 I)>0$ and $\lambda_0 \notin\overline{\sigma_{CI}(T)}$ for the reason of $\sigma_a(T)=\sigma(T)=\sigma_{ea}(T)$. It follows that $\lambda\in\rho_{CI}(T)\cap\rho_{ea}(T)=\rho(T) $ if $0<\lvert \lambda-\lambda_0\rvert$ small enough. Thus $\lambda_0\in E(T)$, which contradicts with `` $E(T)=\emptyset$". The claim II is proved.  
Thus a desired result emerges.
 \end{proof}

 Based on the proof of Proposition 4.2, the transfer of property ($UW${\scriptsize \it{E}}) from $T$ to $f(T)$ can be described as follows.

\begin{thm} Suppose $T\in$($UW${\scriptsize \it{E}}). Then
$f(T)\in$($UW${\scriptsize \it{E}}) for any $f\in$Hol$(\sigma(T))$ if and only if one of the following three cases occurs:

$1)$\   $\sigma_{CI}(T)=\rho_s(T)\cap\sigma(T)$, where $\rho_{s}(T)=\{\lambda\in{\Bbb C}: T-\lambda I$ is surjective$\}$;

$2)$\   $\rho_{SF_+}(T)\cap\sigma(T)\subseteq\sigma_{CI}(T)$ and $\sigma_{CI}(T)=[\rho_{a}(T)\cap\sigma(T)]\cup[\rho_{s}(T)\cap\{\lambda\in{\Bbb C}: n(T-\lambda I)=\infty\}]$;

$3)$\   $\sigma_b(T)=[\overline{\sigma_{CI}(T)}\cap\sigma_{ea}(T)]\cup\hbox{acc}[\rho_{CI}(T)\cap\sigma_{ea}(T)]\cup\hbox{acc}\sigma_K(T)$.
 \end{thm}

So far, the transfer of property ($UW${\scriptsize \it{E}}) from $T$ to $f(T)$ for any $f\in$Hol$(\sigma\\(T))$ has been achieved. As might be expected,
it cannot conclude that both $f(T)$ and  $f(T^{\ast})$ are in ($UW${\scriptsize \it{E}}) for any $f\in$Hol$(\sigma(T))$ simultaneously even if $T$ and $T^{\ast}$ are both in ($UW${\scriptsize \it{E}}), such as the $T_5$ mentioned at the beginning of this section.
Moreover, $f(T)\in$($UW${\scriptsize \it{E}})  for any $f\in$Hol$(\sigma(T))$ and  $f(T^{\ast})\in$($UW${\scriptsize \it{E}}) for any $f\in$Hol$(\sigma(T))$ cannot be extrapolated from each other. For instance, put $T_6\in B(\ell^2)$ be defined by:
$T_6(x_1, x_2, x_3, \cdots)=(0, x_1,0, \frac{x_2}{2}, \cdots)$, then $f(T_6)\in$($UW${\scriptsize \it{E}}) for any  $f\in$Hol$(\sigma(T_6))$ but ${T_6}^{\ast}\notin$($UW${\scriptsize \it{E}}).
The ultimate aim of this section is to achieve the transfer of property ($UW${\scriptsize \it{E}}) from $T$ to $f(T)$ and $f(T^{\ast})$ simultaneously, which also arouses our interest in exploring whether $f(T)$ and $f(T^{\ast})$ are both in ($UW${\scriptsize \it{E}}) for any $f\in$Hol$(\sigma(T))$ simultaneously.

\begin{prop}  Let $T\in B(H)$. Then both $f(T)$ and $f(T^{\ast})$ are in ($UW${\scriptsize \it{E}}) for any $f\in$Hol$(\sigma(T))$ if and only if one of the following four cases occurs:

$1)$\  $\sigma(T)=[\overline{\sigma_{CI}(T)}\cap\sigma_{SF}(T)]\cup\Pi^{\ast}(T)\cup\hbox{acc}\sigma_{SF}(T)\cup[\rho_a(T)\cap\sigma(T)]$;

$2)$\  $\sigma(T)=[\overline{\sigma_{CI}(T)}\cap\sigma_{SF}(T)]\cup\Pi^{\ast}(T)\cup\hbox{acc}\sigma_{SF}(T)\cup[\rho_s(T)\cap\sigma(T)]$;

$3)$\  $\sigma(T)=[\overline{\sigma_{CI}(T)}\cap\sigma_{SF}(T)]\cup\Pi^{\ast}(T)\cup\hbox{acc}\sigma_{SF}(T)\cup[\rho_a(T)\cap\sigma_{SF_-}(T)]\cup[\rho_s(T)\cap\sigma_{SF_+}(T)]$;

$4)$\  $\sigma_b(T)=\hbox{acc}\sigma(T)\cap\sigma_{SF}(T)$.
  
\end{prop}

\begin{proof} There are some facts as follows from Lemma 4.1.

If $\sigma_0(T)\neq\emptyset$, then $f(T)\in$($UW${\scriptsize \it{E}}) for any $f\in$Hol$(\sigma(T))$ if and only if $\sigma_{ea}(T)=\sigma_{b}(T)$ and $\hbox{iso}\sigma(T)=\sigma_0(T)$. Hence $f(T)$ and  $f(T^{\ast})$ are in ($UW${\scriptsize \it{E}}) simultaneously for any $f\in$Hol$(\sigma(T))$ if and only if $\sigma_{b}(T)=\hbox{acc}\sigma(T)\cap\sigma_{SF}(T)$ provided that $\sigma_0(T)\neq\emptyset$.

If $\sigma_0(T)=\emptyset$ and both $f(T)$ and $f(T^{\ast})$ are in ($UW${\scriptsize \it{E}})  for any $f\in$Hol$(\sigma(T))$, then $\sigma_a(T)=\sigma_{ea}(T)$ and $E(T)=\emptyset$ simultaneously with $\sigma_a(T^{\ast})=\sigma_{ea}(T^{\ast})$ and $E(T^{\ast})=\emptyset$. Meanwhile,  
we obtain the following two contents,

 $(1)$\ the product $\hbox{ind}(T-\lambda I)\cdot \hbox{ind}(T-\mu I)\geq 0$ for any $\lambda,\mu\in\rho_{{SF}_{+}}(T)$
 \\and
 
$(2)$\ the product $\hbox{ind}(T-\lambda^{\prime} I)\cdot \hbox{ind}(T-\mu^{\prime} I)\geq 0$ for any $\lambda^{\prime},\mu^{\prime}\in\rho_{{SF}_{-}}(T)$,\\ from which we will make further classifications as follows.
\begin{enumerate}
\item  $\hbox{ind}(T-\lambda I)\leq 0$ for any $\lambda\in\rho_{SF}(T)$.

Since $\sigma_a(T)=\sigma_{ea}(T)$, $T-\lambda I$ is bounded below for any $\lambda\in\rho_{SF}(T)$. It shows that $\rho_{CI}(T)\cap\rho_{SF}(T)=\rho(T)$, and hence $\hbox{acc}\sigma(T)=\overline{\sigma_{CI}(T)}\cup\hbox{acc}\sigma_{SF}(T)$.
Now we claim the following two statements are equivalent:\\
 $(i)$\ $\sigma_0(T)=\emptyset$ and $T^{\ast}$ and $T$ are both in ($UW${\scriptsize \it{E}}) with $\hbox{ind}(T-\lambda I)\leq 0$ for any $\lambda\in\rho_{SF}(T)$; \\
 $(ii)$\ $\sigma(T)=[\overline{\sigma_{CI}(T)}\cap\sigma_{SF}(T)]\cup\Pi^{\ast}(T)\cup\hbox{acc}\sigma_{SF}(T)\cup[\rho_a(T)\cap\sigma(T)]$.

Indeed, assume that $\sigma_0(T)=\emptyset$ and $T^{\ast}$ and $T$ are in ($UW${\scriptsize \it{E}}), then $\sigma(T)=\Pi^{\ast}(T)\cup\hbox{acc}\sigma(T)$. Besides, note that  $\hbox{acc}\sigma(T)=\overline{\sigma_{CI}(T)}\cup\hbox{acc}\sigma_{SF}(T)$, we can deduce that
$\sigma(T)\subseteq\Pi^{\ast}(T)\cup\hbox{acc}\sigma_{SF}(T)\cup[\overline{\sigma_{CI}(T)}\cap\sigma_{SF}(T)]\cup[\sigma(T)\cap\rho_{a}(T)]$. The other inclusion is apparent.
Conversely, it concludes that $\big\{[\sigma_a(T)\backslash\sigma_{ea}(T)]\cup E(T)\big\}\cap\sigma(T)=\emptyset$ for the reason of $\big\{[\sigma_a(T)\backslash\sigma_{ea}(T)]\cup E(T)\big\}\cap\big\{[\overline{\sigma_{CI}(T)}\cap\sigma_{SF}(T)]\cup\Pi^{\ast}(T)\cup\hbox{acc}\sigma_{SF}(T)\cup[\rho_a(T)\cap\sigma(T)]\big\}=\emptyset$.
Thus $T\in$($UW${\scriptsize \it{E}}) and $\sigma_0(T)=\emptyset$.
Similarly, one can also see that $T^{\ast}$ are in ($UW${\scriptsize \it{E}}). The rest is straightforward to verify and thereby omitted. The claim is proved.

\item $\hbox{ind}(T-\lambda I)\geq 0$ for any $\lambda\in\rho_{SF}(T)$.

Note that $\sigma_a(T^{\ast})=\sigma_{ea}(T^{\ast})$. Then $T-\lambda I$ is surjective, namely, $\rho_{SF}(T)=\rho_s(T)$. Similar to the case 1 above, we deduce that $\hbox{acc}\sigma(T)=\overline{\sigma_{CI}(T)}\cup\hbox{acc}\sigma_{SF}(T)$.
Now it suffices to check the following two statements are equivalent:\\
  $(i)$\ $\sigma_0(T)=\emptyset$ and $T^{\ast}$ and $T$ are in ($UW${\scriptsize \it{E}}) with $\hbox{ind}(T-\lambda I)\geq 0$ for any $\lambda\in\rho_{SF}(T)$;\\
  $(ii)$\ $\sigma(T)=[\overline{\sigma_{CI}(T)}\cap\sigma_{SF}(T)]\cup\Pi^{\ast}(T)\cup\hbox{acc}\sigma_{SF}(T)\cup[\rho_s(T)\cap\sigma(T)]$.
The proof is similar to the claim in case 1 above, and thus we omit it.

\item  $\hbox{ind}(T-\lambda I)\leq 0$ for any $\lambda\in\rho_{SF_+}(T)$, and
$\hbox{ind}(T-\mu I)\geq 0$ for any $\mu\in\rho_{SF_-}(T)$.

 Under this circumstance, we have three statements as follows:
 $(a)$\ $\rho_{SF_+}(T)=\rho(T)\cup[\rho_a(T)\cap\sigma_{SF_-}(T)]$ and $\rho_{SF_-}(T)=\rho(T)\cup[\rho_s(T)\cap\sigma_{SF_+}(T)]$;
  $(b)$\  $\sigma_{CI}(T)=[\rho_a(T)\cap\sigma_{SF_-}(T)]\cup[\rho_s(T)\cap\sigma_{SF_+}(T)]$;
  $(c)$\  $\hbox{acc}\sigma(T)=\overline{\sigma_{CI}(T)}\cup\hbox{acc}\sigma_{SF}(T)$.
Now it suffices to prove the following two statements are equivalent:\\
$(i)$\  $\sigma_0(T)=\emptyset$ and $T^{\ast}$ and $T$ are in ($UW${\scriptsize \it{E}}) with $\hbox{ind}(T-\lambda I)\leq 0$ for any $\lambda\in\rho_{SF_+}(T)$ and
$\hbox{ind}(T-\mu I)\geq 0$ for any $\mu\in\rho_{SF_-}(T)$;\\
 $(ii)$\ $\sigma(T)=[\overline{\sigma_{CI}(T)}\cap\sigma_{SF}(T)]\cup\Pi^{\ast}(T)\cup\hbox{acc}\sigma_{SF}(T)\cup[\rho_a(T)\cap\sigma_{SF_-}(T)]\cup[\rho_s(T)\cap\sigma_{SF_+}(T)]$.\\
Similar to the claim in case 1, the proof is not presented here.

\item $\hbox{ind}(T-\lambda I)\geq 0$ for any $\lambda\in\rho_{SF_+}(T)$,
and $\hbox{ind}(T-\mu I)\leq 0$ for any $\mu\in\rho_{SF_-}(T)$.

  It is clear to check that the two sets $\{\lambda\in\rho_{SF_+}(T): \hbox{ind}(T-\lambda I)> 0\}$ and $\{\mu\in\rho_{SF_-}(T): \hbox{ind}(T-\mu I)<0\}$ are both empty, and hence $\hbox{ind}(T-\lambda I)=0$ for any $\lambda\in\rho_{SF}(T)$.
 This may boil down to either case 1 or case 2.
 \end{enumerate}
 The proof is completed.
 \end{proof}

\begin{rem}Using Proposition 4.4 one can construct an operator $T\in B(H)$ such that $f(T)$ and $f({T}^{\ast})$ are both in ($UW${\scriptsize \it{E}}) whenever $f\in$Hol$(\sigma(T))$. Indeed, let $T_7\in B(\ell^2)$ be defined by:
$T_7(x_1, x_2, x_3, \cdots)=( x_1, \frac{x_2}{2},\frac{x_3}{3},\cdots)$. One has $\sigma_b(T_7)=\{0\}=\hbox{acc}\sigma(T_7)\cap\sigma_{SF}(T_7)$, and hence $f(T_7)$ and $f({T_7}^{\ast})$ are both in ($UW${\scriptsize \it{E}}) for any  $f\in$Hol$(\sigma(T_7))$.
\end{rem}

From the proof of Proposition 4.4, we can characterize the transfer of property ($UW${\scriptsize \it{E}}) from $T$ to $f(T)$ and $f(T^{\ast})$ simultaneously as follows.

\begin{thm} Suppose $T\in$($UW${\scriptsize \it{E}}). Then both $f(T)$ and $f(T^{\ast})$ are in ($UW${\scriptsize \it{E}}) for any $f\in$Hol$(\sigma(T))$ if and only if one of the following two cases occurs:

$1)$\   $\sigma(T)=\sigma_{CI}(T)\cup\sigma_{SF}(T)$ and $\sigma(T)=\Pi^{\ast}(T)\cup\hbox{acc}\sigma(T)$  with one of the following three cases holds:

\quad$i)$\   $\sigma_{CI}(T)=\rho_a(T)\cap\sigma(T)$;

\quad$ii)$\   $\sigma_{CI}(T)=\rho_s(T)\cap\sigma(T)$;

 \quad$iii)$\  $\sigma_{CI}(T)=[\rho_a(T)\cap\sigma_{SF_-}(T)]\cup[\rho_s(T)\cap\sigma_{SF_+}(T)]$.

$2)$\   $\sigma_b(T)=\hbox{acc}\sigma(T)\cap\sigma_{SF}(T)$.

 \end{thm}

\section{Entanglement stability of property ($UW_E$) for $2\times2$ upper triangular operator matrices}

The aim of the section is to investigate the entanglement stability of property ($UW${\scriptsize \it{E}}) for $2\times2$ upper triangular operator matrices. 
The study arises naturally from the following fact: If $T\in B(\mathcal{H})$ and $M\subset \mathcal{H}$ is a closed subspace, then $T$ can be written as a $2\times 2$ matrix with operator entries,
$$ T=\left(
    \begin{array}{cc}
      A & C \\
      D & B \\
    \end{array}
  \right): M\oplus M^{\perp}\rightarrow  M\oplus M^{\perp}.
$$
If $M$ is invariant under $T$, $T$ has an upper triangular operator matrix representation. As a result, the study for $T$ can be reduced to the study of $A,B$ and $C$.
Now put $M=H$ and $M^{\perp}=K$, where $H$ and $K$ are infinite dimensional complex separable Hilbert spaces.
From now on, we always suppose that $A\in B(H)$ and $B\in B(K)$, and by $M_C$ we denote  an operator acting on $H\oplus K$ with the form
$$M_C=\left(
    \begin{array}{cc}
      A & C \\
      0 & B \\
    \end{array}
  \right),
$$
where $C\in B(K,H).$ Besides, write $M_C$ as $M_0$ if $C=0$.

For a given operator pair $(A, B)\in(B(H),B(K))$, if $M_0\in$($UW${\scriptsize \it{E}}), it's natural to consider whether $M_C\in$($UW${\scriptsize \it{E}}) for any $C\in B(K,H).$ 
In \cite{af}, we has conducted such a research,
the essence of which is to influence $C$ and ultimately affect the whole $M_C$ through the mutual entanglement between $A$ and $B$.
Now we continue to characterize the entanglement stability of property ($UW${\scriptsize \it{E}}) for $M_C$ with the help of the so-called $(A,B)$ entanglement stable spectra analogous to \cite{ff}. Some useful lemmas are presented first.

\begin{lem} Suppose $(A, B)\in(B(H),B(K))$. If
$A$ is upper semi-Fredholm and Kato (that is, $N(A)\subseteq\bigcap\limits_{n=1}^\infty R(A^n)$) with $d(A)=\infty$, then there is some $C\in B(K,H)$ such that $M_C$ is upper semi-Fredholm and Kato with $\hbox{ind}M_C<0$.
\end{lem}

\begin{proof} If $B$ is upper semi-Fredholm, then $M_C$ is upper semi-Fredholm with $\hbox{ind}M_C=-\infty$ for any $C\in B(K,H)$. Since $\hbox{dim}R(A)^{\perp}=\infty$,  there is an invertible isometry $T: K\rightarrow R(A)^{\perp}$.

Define an operator $C_0$ by
 $$C_0=\left(
         \begin{array}{c}
           T \\
           0 \\
         \end{array}
       \right):K\rightarrow\left(
                                  \begin{array}{c}
                                    R(A)^{\perp} \\
                                    R(A) \\
                                  \end{array}
                                \right).$$
It is easy to see that $N(M_{C_0})=N(A)\oplus\{0\}\subseteq\bigcap\limits_{n=1}^\infty R(A^n)\oplus\{0\}\subseteq\bigcap\limits_{n=1}^\infty R({M_{C_0}}^n)$.
Hence $M_{C_0}$ is upper semi-Fredholm and Kato with $\hbox{ind}M_{C_0}< 0.$

If $B$ is not upper semi-Fredholm, $R(B)$ is either closed with $n(B)=\infty$ or not closed.

 Assume firstly that $R(B)$ is not closed.
Set $C=C_0$. Then $M_{C_0}$ is upper semi-Fredholm and $N(M_{C_0})=N(A)\oplus\{0\}$ and thus $M_{C_0}$ is Kato. Besides, in view of $d(M_{C_0})\geq d(B)$, one has $d(M_{C_0})=\infty$, and hence $\hbox{ind}M_{C_0}< 0.$

 Assume secondly that $R(B)$ is closed with $n(B)=\infty$.
Now we have either $n(A)<d(B)$ or $n(A)\geq d(B)$.
\begin{enumerate} 
\item Suppose $n(A)<d(B)$.

Since $N(B)$ and $R(A)^{\perp}$ are both infinite dimensional, there is an invertible isomorphism $T: N(B)\rightarrow R(A)^{\perp}$.

Define an operator $C:K\rightarrow H$ by
 $$C=\left(
         \begin{array}{cc}
           T & 0 \\
           0 & 0 \\
         \end{array}
       \right):\left(
                \begin{array}{c}
                 N(B)\\
                  N(B)^{\perp} \\
                \end{array}
              \right)\rightarrow\left(
                                  \begin{array}{c}
                                    R(A)^{\perp} \\
                                    R(A) \\
                                  \end{array}
                                \right).
$$
One can see that $N(M_C)=N(A)\oplus\{0\}$. This shows $M_C$ is Kato and $n(M_C)<\infty$. Now, it suffices to show that $M_C$ is upper semi-Fredholm with $\hbox{ind}M_C\leq 0$. In other words, we will show that $R(M_C)$ is closed and $\hbox{ind}M_C<0$.

\begin{enumerate} 
\item $R(M_C)$ is closed.

Suppose $\bigl(
                          \begin{smallmatrix}
                                        u_0 \\
                                        v_0 \\
                                      \end{smallmatrix}\bigr)
                                    \in\overline{R(M_C)}$, 
there is a sequence $\Big\{\bigl(
                          \begin{smallmatrix}
                                        u_n \\
                                        v_n \\
                                      \end{smallmatrix}\bigr)\Big\}$ such that
$M_C\bigl(
                          \begin{smallmatrix}
                                        u_n \\
                                        v_n \\
                                      \end{smallmatrix}\bigr)$
                                     $\rightarrow \bigl(
                          \begin{smallmatrix}
                                        u_0 \\
                                        v_0 \\
                                      \end{smallmatrix}\bigr)
$($n\rightarrow\infty$). Namely, $Au_n +Cv_n\rightarrow u_0$ and $Bv_n\rightarrow v_0$ ($n\rightarrow\infty$). Note that $\{Au_n\}\subseteq R(A)$ and $\{Cv_n\}\subseteq R(A)^\perp$. Then $\{Au_n\}$ and $\{Cv_n\}$ are Cauchy sequences. Put $v_n=\alpha_n +\beta_n$, where $\alpha_n\in N(B)$ and $\beta_n\in N(B)^{\perp}$. Then $Cv_n=T\alpha_n$. Besides, from the invertibility of $T$,  we deduce that $\{\alpha_n\}$ is a Cauchy sequence. Thus $\{B\alpha_n\}$ is a Cauchy sequence. This implies $\{B\beta_n\}$ is a Cauchy sequence since $B\beta_n=Bv_n-B\alpha_n$. Furthermore,  $B\mid_{ N(B)^{\perp}}$ is bounded below, it shows $\{\beta_n\}$ is a Cauchy sequence. Hence $\{v_n\}$ is a Cauchy sequence, and there is some $y_0$ such that  $v_n\rightarrow y_0$ ($n\rightarrow\infty$). Then $Cv_n\rightarrow Cy_0$ and $Bv_n\rightarrow By_0$.
Since $R(A)$ is closed, there is some $x_0$ such that $Au_n\rightarrow Ax_0$ ($n\rightarrow\infty$). Then $M_C\bigl(
                          \begin{smallmatrix}
                                        x_0 \\
                                        y_0 \\
                                      \end{smallmatrix}\bigr)=\bigl(
                          \begin{smallmatrix}
                                        u_0 \\
                                        v_0 \\
                                      \end{smallmatrix}\bigr),$
and $R(M_C)$ is closed.

\item  $\hbox{ind}M_C<0$.

Suppose on the contrary that $\hbox{ind}M_C\geq0$, that is, $n(M_C)\geq d(M_C)$. Combining $n(M_C)=n(A)$ and $d(M_C)\geq d(B)$, we deduce $n(A)\geq d(B)$, which is a contradiction.
\end{enumerate} 

\item Suppose $n(A)\geq d(B)$.

Since $A$ is upper semi-Fredholm, $B$ is lower semi-Fredholm. Write $n(A)=n$ and $d(B)=m$. $B$ is not a compact operator seeing that $B$ is lower semi-Fredholm. Then $R(B^\ast)=N(B)^{\perp}$ is infinite dimensional.
Suppose $N(B)^{\perp}=N\oplus M$, where $\hbox{dim}M=n-m+1$ (this guarantees $n-m<\hbox{dim}M<\infty$),
and $R(A)^{\perp}=E\oplus F$, where $\hbox{dim}F=\hbox{dim}M$. Then $\hbox{dim}E=\infty$. From the preceding procedure, there are two invertible isometries  $T_1:N(B)\rightarrow E$ and $T_2:M\rightarrow F$.

Define an operator $C:K\rightarrow H$ by
 $$C=\left(
         \begin{array}{ccc}
           T_1 & 0 & 0\\
           0 & T_2 & 0 \\
           0 & 0 & 0 \\
         \end{array}
       \right):\left(
                \begin{array}{c}
                  N(B) \\
                  M \\
                  N \\
                \end{array}
              \right)\rightarrow\left(
                                  \begin{array}{c}
                                    E \\
                                    F \\
                                    R(A)\\
                                  \end{array}
                                \right).$$
It is straightforward to check that $N(M_C)=N(A)\oplus\{0\}$. Hence $M_C$ is Kato and $n(M_C)<\infty$.
Now we show that $M_C$ is upper semi-Fredholm with $\hbox{ind}M_C\leq 0$.
\begin{enumerate} 
\item $R(M_C)$ is closed.

Assume $\bigl(
                          \begin{smallmatrix}
                                        u_0 \\
                                        v_0 \\
                                      \end{smallmatrix}\bigr)\in\overline{R(M_C)}$. Then there is a sequence $\Big\{\bigl(
                          \begin{smallmatrix}
                                        x_n \\
                                        y_n \\
                                      \end{smallmatrix}\bigr)\Big\}$ such that
     $M_C\bigl(
                          \begin{smallmatrix}
                                        x_n \\
                                        y_n \\
                                      \end{smallmatrix}\bigr)\rightarrow \bigl(
                          \begin{smallmatrix}
                                        u_n \\
                                        v_n \\
                                      \end{smallmatrix}\bigr)
$($n\rightarrow\infty$), i.e., $Ax_n +Cy_n\rightarrow u_0$ and $By_n\rightarrow v_0$ ($n\rightarrow\infty$). By the definition of the operator $C$, $\{Ax_n\}$ and $\{Cy_n\}$ are Cauchy sequences.
There is some $x_0$ such that $Ax_n\rightarrow Ax_0$ ($n\rightarrow\infty$) since $R(A)$ is closed. Let $y_n=\alpha_n +\beta_n+\gamma_n$, where $\alpha_n\in N(B)$, $\beta_n\in M$ and $\gamma_n\in N$. Then $Cy_n=T_1\alpha_n+T_2\beta_n=\bigl(
                          \begin{smallmatrix}
                                        T_1 & 0 \\
                                        0 & T_2 \\
                                      \end{smallmatrix}\bigr)\bigl(
                          \begin{smallmatrix}
                                        \alpha_n \\
                                        \beta_n \\
                                      \end{smallmatrix}\bigr).$
For $\bigl(
                          \begin{smallmatrix}
                                        T_1 & 0 \\
                                        0 & T_2 \\
                                      \end{smallmatrix}\bigr)$ is invertible, $\{\alpha_n +\beta_n\}$ is a Cauchy sequence. Moreover, from $\alpha_n\in N(B)$ and $\beta_n\in M\subseteq N(B)^\perp$, it follows that $\{\alpha_n\}$ and $\{\beta_n\}$ are both Cauchy sequences. Note that $By_n=B(\beta_n+\gamma_n)$ and  $B\mid_{ N(B)^{\perp}}$ is bounded below, we deduce that  $\{\beta_n+\gamma_n\}$ is a Cauchy sequence. Then $\{\gamma_n\}$ and furthermore, $\{y_n\}$ is a Cauchy sequence. Suppose $y_0$ is the limit point of $\{y_n\}$, thereby $Cy_n\rightarrow Cy_0$ and $By_n\rightarrow By_0$ ($n\rightarrow\infty$). Thus $M_C\bigl(
                          \begin{smallmatrix}
                                        x_0 \\
                                        y_0 \\
                                      \end{smallmatrix}\bigr)=\bigl(
                          \begin{smallmatrix}
                                        u_0 \\
                                        v_0 \\
                                      \end{smallmatrix}\bigr)$, and $\bigl(
                          \begin{smallmatrix}
                                        u_0 \\
                                        v_0 \\
                                      \end{smallmatrix}\bigr)\in R(M_C)$.

\item   $\hbox{ind}M_C<0$.

It is easy to see that $\{0\}\oplus N(B^\ast)\subseteq N({M_C}^\ast)$ and $n(B^\ast)=d(B)=m$.
Suppose that $\{e_1, e_2, \cdots,e_m\}$ is an orthonormal basis in $N(B^\ast)$. Note that $F\subseteq R(A)^\perp=N(A^\ast)$ and $\hbox{dim}F=n-m+1$. Then let $\{x_1, x_2, \cdots,x_{n-m+1}\}$ be an orthonormal basis in $F$. Furthermore, by invertibility of ${T_2}^\ast$,
$\{{T_2}^\ast x_1,{T_2}^\ast x_2, \cdots,{T_2}^\ast x_{n-m+1}\}$ is a linearly independent subset of $M$. Since $M\subseteq N(B)^\perp=R(B^\ast)$, there are $y_1, y_2, \cdots,y_{n-m+1}$ in $Y$ such that ${T_2}^\ast x_i=B^\ast y_i$, where $i=1,2,\cdots,n-m+1.$ Hence $\bigl(
                          \begin{smallmatrix}
                                        x_i \\
                                        -y_i \\
                                      \end{smallmatrix}\bigr)\in N({M_C}^\ast)$.
It is easy to see that $\Big\{ \bigl(
                          \begin{smallmatrix}
                                        x_1 \\
                                        -y_1 \\
                                      \end{smallmatrix}\bigr), \bigl(
                          \begin{smallmatrix}
                                        x_2 \\
                                        -y_2 \\
                                      \end{smallmatrix}\bigr), \cdots,\bigl(
                          \begin{smallmatrix}
                                        x_{n-m+1} \\
                                        -y_{n-m+1} \\
                                      \end{smallmatrix}\bigr),\bigl(
                          \begin{smallmatrix}
                                         0 \\
                                         e_1 \\
                                      \end{smallmatrix}\bigr), \cdots,\bigl(
                          \begin{smallmatrix}
                                          0 \\
                                          e_m \\
                                      \end{smallmatrix}\bigr)\Big\}$ is linearly independent in $N({M_C}^\ast)$. This shows $n({M_C}^\ast)\geq n-m+1+m=n+1$ and hence $d(M_C)\geq n+1.$ Thus $\hbox{ind}M_C<0$ since $n(M_C)=n$. 
\end{enumerate} 
\end{enumerate}  
The proof is completed.
\end{proof}

From the proof of Lemma 5.1, some results emerged as follows.

 Let $A\in B(H)$ be upper semi-Fredholm and Kato with $d(A)=\infty$, and $B\in B(K)$. The following statements hold:
 \begin{enumerate}
   \item Suppose $n(B)=\infty$. If $n(A)>d(B)$, there is some $C\in B(K,H)$ such that $M_C$ is upper semi-Fredholm and Kato with $n(M_C)>0$ and $\hbox{ind}M_C<0$.
   \item Suppose $n(A)>0$. There is some $C\in B(K,H)$ such that $M_C$ is upper semi-Fredholm and Kato with $n(M_C)>0$ and $\hbox{ind}M_C<0$.
 \end{enumerate}

The following lemma is useful to accomplish the aim in this section. 

\begin{lem} \cite[Lemma 2.1]{am} Suppose $(A, B)\in(B(H),B(K))$. If $A$ is upper semi-Fredholm with $\hbox{asc}(A)<\infty$, $d(A)=\infty$ and $n(A)+n(B)>0$ (abbrev. $N_0(A; B)>0$), then there exists $C\in B(K,H)$ such that $M_C$ is an upper semi-Fredholm operator with $n(M_C)>0$, $\hbox{ind}M_C<0$ and $\hbox{asc}(M_C)<\infty$.
\end{lem}

In the following, we abbreviate this notation ``$n(A-\lambda I)+n(B-\lambda I)$" (resp. ``$d(A-\lambda I)+d(B-\lambda I)$") as ``$N_{\lambda}(A; B)$" (resp. ``$D_{\lambda}(A; B)$") for the sake of simplicity.
Now it's time to introduce the concepts of the $(A,B)$ entanglement stable spectra according to Lemma 5.1 and Lemma 5.2. 
\begin{defn}
For a given pair $(A,B)\in (B(H),B(K))$, define\\ 
 $\Omega_1(A; B)=\{ \lambda\in\rho_{ab}(A): d(A-\lambda I)=\infty$ and $N_{\lambda}(A; B)>0\}$;\\ $\Omega_2(A; B)=\{\lambda\in\rho_{SF_+}(A)\cap\rho_{K}(A): d(A-\lambda I)=\infty$ and $n(A-\lambda I)>0\}$; \\$\Omega_3(A; B)=\{\lambda\in\sigma_{CI}(A): n(A-\lambda I)<\infty, n(B-\lambda I)\leq d(A-\lambda I)$ and 

\qquad\quad\ \,$N_{\lambda}(A; B)= D_{\lambda}(A; B)\}$;\\ $\Omega(A; B)=\{\lambda\in\rho_{SF_+}(A): d(A-\lambda I)=\infty$ and $N_{\lambda}(A; B)>0\}$, \\which are collectively referred to as the $(A,B)$ entanglement stable spectra. 
\end{defn}

With the help of the entanglement stable spectra,  let's discuss the circumstances in which $M_C\in$($UW${\scriptsize \it{E}}) for any $C\in B(K, H)$.

\begin{thm} For an operator pair $(A, B)\in(B(H),B(K))$, $M_C\in$($UW${\scriptsize \it{E}}) for any $C\in B(K, H)$ if and only if the following conditions hold:

$1)$\  $M_0\in$($UW${\scriptsize \it{E}});

$2)$\  $\Omega_1(A; B)\cup\Omega_2(A; B)=\emptyset$.
  \end{thm}

\begin{proof} With the help of Lemma 5.2, it suffices to show $\Omega_2(A; B)=\emptyset$ for the necessity. Assume on the contrary that $\lambda_0\in\Omega_2(A; B)$, then, from Lemma 5.1, there is some $C$ such that $\lambda_0\in\sigma_a(M_C)\backslash\sigma_{ea}(M_C)$ and $M_C-\lambda_0 I$ is Kato. Besides, in view of $M_C\in$($UW${\scriptsize \it{E}}), $M_C-\lambda_0 I$ is Browder.  Hence, by \cite[Lemma 3.4]{ah}, $M_C-\lambda_0 I$ is invertible, which contradicts with $\lambda_0\in\sigma_a(M_C)$.

Now we show the sufficiency. For any $C\in B(K,H)$,
suppose $\lambda_0\in\sigma_a(M_C)\setminus\sigma_{ea}(M_C)$, then $n(A-\lambda_0I)<\infty$ and $N_{\lambda_0}(A; B)>0$.
If $d(A-\lambda_0I)=\infty$, then $A-\lambda I$ is upper semi-Fredholm and Kato with $d(A-\lambda I)=\infty$ whenever $0<\lvert \lambda-\lambda_0\rvert$ is small enough.
Since $\Omega_2(A; B)=\emptyset$, $A-\lambda I$ is bounded below. This implies $\hbox{asc}(A-\lambda_0 I)<\infty$, which contradicts with $\Omega_1(A; B)=\emptyset$. Therefore $d(A-\lambda_0I)<\infty$, that is, $A-\lambda_0I$ is Fredholm. It follows that $B-\lambda_0 I$ is upper semi-Fredholm. Thus  $\lambda_0\in\sigma_a(M_0)\setminus\sigma_{ea}(M_0)$. Combining with $M_0\in$($UW${\scriptsize \it{E}}), we can see that $\lambda_0\in\sigma_0(M_C)\subseteq E(M_C)$.
Conversely, suppose $\lambda_0\in E(M_C)$. Then $N_{\lambda_0}(A; B)>0$. Besides, we have that $A-\lambda I$ is bounded below, $B-\lambda I$ is surjective and $d(A-\lambda I)=n(B-\lambda I)$ whenever $0<\lvert \lambda-\lambda_0\rvert$ is small enough. From $\Omega_1(A; B)=\emptyset$, we can deduce $d(A-\lambda I)=n(B-\lambda I)<\infty$. Then $M_0-\lambda I$ is Weyl, and furthermore Browder, which infers that  $A-\lambda I$ and $B-\lambda I$ are both invertible. That is to say $\lambda_0\in\hbox{iso}\sigma(M_0)$. Note that $N_{\lambda_0}(A; B)>0$, then $\lambda_0\in E(M_0)$. It follows from $M_0\in$($UW${\scriptsize \it{E}}) that $M_C-\lambda_0 I$ is Browder. This completes the proof.
\end{proof}

\begin{cor} For an operator pair $(A, B)\in(B(H),B(K))$, $M_C\in$($UW${\scriptsize \it{E}}) for any $C\in B(K, H)$ if and only if the following conditions hold:

$1)$\  $M_0\in$($UW${\scriptsize \it{E}});

$2)$\  If there is some $C\in B(K, H)$ such that $\sigma_a(M_C)\backslash\sigma_{ab}(M_C)\neq\emptyset$, then  $\sigma_a(M_C)\backslash\sigma_{ab}(M_C)=\sigma_0(M_C)$;

$3)$\  $\Omega_2(A; B)\cup\Omega_3(A; B)=\emptyset$.

\end{cor}

\begin{proof} In virtue of Theorem 5.4, it suffices to show the set $\Omega_3(A; B)$ is empty for the necessity.
Suppose otherwise that $\lambda_0\in\Omega_3(A; B)$, then there are two cases as follows.

Assume firstly that $A-\lambda_0 I$ is bounded below but not surjective.
Note that $n(B-\lambda_0 I)=D_{\lambda_0}(A; B)$ and $n(B-\lambda_0 I)\leq d(A-\lambda_0 I)$, then $d(B-\lambda_0 I)=0$ and $n(B-\lambda_0 I)= d(A-\lambda_0 I)$. According to $\Omega_1(A; B)=\emptyset$, $n(B-\lambda_0 I)\leq d(A-\lambda_0 I)<\infty$. It shows that $M_0-\lambda_0 I$ is Weyl. One can see $A-\lambda_0 I$ and $B-\lambda_0 I$ are Browder since $M_0\in$($UW${\scriptsize \it{E}}). Thus $A-\lambda_0 I$ is invertible, which contradicts with $\lambda_0\in\sigma_{CI}(A)$.

Assume secondly that $A-\lambda_0 I$ is surjective but not injective.
We have $n(B-\lambda_0 I)=0$ and $n(A-\lambda_0 I)= d(B-\lambda_0 I)<\infty$.
It concludes that $A-\lambda_0 I$ and $B-\lambda_0 I$ are both Fredholm and $M_0-\lambda_0 I$ is Weyl. For $M_0\in$($UW${\scriptsize \it{E}}), $A-\lambda_0 I$ is Browder, and furthermore invertible, a contradiction.

For the sufficiency, it suffices to show the set $\Omega_1(A; B)$ is empty.
Otherwise, assume $\lambda_0\in\Omega_1(A; B)$. By means of Lemma 5.2, there is some $C$ such that $\lambda_0\in\sigma_a(M_C)\backslash\sigma_{ab}(M_C)$ and $\hbox{ind}(M_C-\lambda_0 I)<0$. And thus, $M_C-\lambda_0 I$ is Browder for condition 2), a desired contradiction.
Thus a desired result emerges.
\end{proof}

A smarter conclusion can be obtained as follows.

\begin{thm} For an operator pair $(A, B)\in(B(H),B(K))$, $M_C\in$($UW${\scriptsize \it{E}}) for any $C\in B(K, H)$ if and only if the following conditions hold:

$1)$\ $M_0\in$($UW${\scriptsize \it{E}});

$2)$\ $\Omega(A; B)=\emptyset$.
 \end{thm}

\begin{proof} 
It suffices to show the necessity from Theorem 5.4.
Assume on the contrary that $\lambda_0\in\Omega(A; B)$. It follows that $A-\lambda I$ is upper semi-Fredholm and Kato with $d(A-\lambda I)=\infty$ whenever  $0<\lvert \lambda-\lambda_0\rvert$ is small enough. In view of Theorem 5.4, $n(A-\lambda I)=0$, hence $\hbox{asc}(A-\lambda_0 I)<\infty$, which contradicts with $\Omega_1(A; B)=\emptyset$.
The proof is completed.
\end{proof}

\begin{rem} Given $(A,B)\in (B(H), B(K))$, it is light to check whether $M_C\in$($UW${\scriptsize \it{E}}) whenever $C\in B(K, H)$:

$(i)$\ Let $A, B\in B(\ell^2)$ be defined by:
$A(x_1, x_2, x_3, \cdots)=( 0, \frac{x_2}{2},\frac{x_3}{3},\cdots)$ and $B(x_1, x_2, x_3, \cdots)=( 0, x_1, \frac{x_2}{2},\frac{x_3}{3},\cdots)$. One can see $\sigma_a(M_0)\setminus\sigma_{ea}(M_0)=\{\frac{1}{2},\frac{1}{3},\cdots\}=E(M_0)$ and $\Omega(A; B)=\emptyset$. Using Theorem 5.6 one can see
$M_C\in$($UW${\scriptsize \it{E}}) for any $C\in B(\ell^2)$. 

$(ii)$\ Let $A, B\in B(\ell^2)$ be defined by:
$A(x_1, x_2, x_3, \cdots)=(0, x_{1}, 0, x_{2}, \cdots)$ and $B(x_1, x_2, x_3, \cdots)=(x_{2}, x_{4}, x_{6}, \cdots)$. It's clear that  $M_0\in$($UW${\scriptsize \it{E}}) whereas $\Omega(A; B)=\{\lambda\in{\Bbb C}: \lvert \lambda \rvert <1\}$, which follows that  $M_{C_0}\notin$($UW${\scriptsize \it{E}}) for some $C_0\in B(\ell^2)$.
  
 \end{rem}

\subsection*{CRediT authorship contribution statement}

\textbf{Sinan Qiu:} Conceptualization, Formal analysis, Validation, Writing-original draft, Writing-review and editing.
\textbf{Lining Jiang:} Conceptualization, Formal analysis, Validation, Writing-original draft, Writing-review and editing.

\subsection*{Declaration of competing interest}

The authors declare that they have no known competing financial interests or personal
relationships that could have appeared to influence the work reported in this paper.

\subsection*{Data availability}

No data was used for the research described in the article.

\subsection*{Acknowledgments}
We are grateful to the referees for helpful comments concerning this paper.

\end{document}